\documentclass[11pt, oneside]{article} 
\usepackage[left=2cm, right=2cm, top=2cm, bottom=2cm]{geometry}
\usepackage[parfill]{parskip}    		
\usepackage{graphicx}			
\usepackage[utf8]{inputenc}
\usepackage[english]{babel} 
\usepackage{amsmath,amssymb,amsthm,bm} 
\usepackage[color=yellow]{todonotes}
\usepackage{mathrsfs}
\usepackage{url}	
\usepackage{esint}
\usepackage[colorlinks]{hyperref}
\usepackage{enumerate}
\usepackage{outlines}[enumerate]
\usepackage{framed,comment,enumerate}
\usepackage{upgreek}
\usepackage{pgfplots}
\usepackage{float}
\usepackage{tikz,tikz-cd}
\usepackage{rotating}
\usepackage{graphicx}
\usepackage{caption}
\usepackage{multicol}
\usepackage{cancel}
\usepackage{subcaption}
\usepackage[title]{appendix}
\usepackage{tikz-cd}
\usepackage{pgfgantt}
\usepackage{rotating}
\usepackage{natbib}

\extrafloats{100}
\numberwithin{equation}{section}

\newtheorem{proposition}{Proposition}[section]

\newtheorem{example}{Example}[section]

\theoremstyle{definition}

\DeclareFontFamily{U}{MnSymbolC}{}
\DeclareSymbolFont{MnSyC}{U}{MnSymbolC}{m}{n}
\DeclareFontShape{U}{MnSymbolC}{m}{n}{
    <-6>  MnSymbolC5
   <6-7>  MnSymbolC6
   <7-8>  MnSymbolC7
   <8-9>  MnSymbolC8
   <9-10> MnSymbolC9
  <10-12> MnSymbolC10
  <12->   MnSymbolC12}{}
\DeclareMathSymbol{\intprod}{\mathbin}{MnSyC}{'270}

\newcommand{\ad}{\operatorname{ad}}

\pgfplotsset{compat=1.16}

\newcommand{\dt}{\mbox{d}}

\begin{document}

\title{\textbf{A comparative numerical study of stochastic Hamiltonian Camassa-Holm equations}}

\author{Darryl D. Holm\thanks{Department of Mathematics, Imperial College London} , Maneesh Kumar Singh\footnotemark[1] \footnote{Corresponding author: maneesh-kumar.singh@imperial.ac.uk} , and Oliver D. Street\thanks{Grantham Institute, Imperial College London} \\ 
\footnotesize
d.holm@ic.ac.uk, maneesh-kumar.singh@imperial.ac.uk, o.street18@imperial.ac.uk 
\\  \small
}
\date{\today}

\maketitle




\maketitle



\begin{abstract}
We introduce a stochastic perturbation of the Camassa-Holm equation such that, unlike previous formulations, energy is conserved by the stochastic flow. We compare this to a complementary approach which preserves Casimirs of the Poisson bracket. Through an energy preserving numerical implementation of the model, we study the influence of noise on the well-known `peakon' formation behaviour of the solution. The energy conserving stochastic approach generates an ensemble of solutions which are spread around the deterministic Camassa-Holm solution, whereas the Casimir conserving alternative develops peakons which may propagate away from the deterministic solution more dramatically. 
\end{abstract}





\section{Introduction}\label{sec-1}

The word `\textit{peakon}' had not yet appeared in the literature as a scientific term before the publication of the paper \citep{camassa1993integrable}. However, the success of that paper has resulted in the words `peakon' and `Camassa--Holm (CH) equation' becoming common terms in the literature of emergent singular solutions of nonlinear wave equations.

An interest in stochastic equations governing the dynamics of a physical system, derived using the geometric structure of the corresponding deterministic model, has recently been developing. This interest has grown in popularity following the inclusion of transport noise in ideal fluid dynamics, developed by \cite{H2015} using a Lie group invariant variational principle. As illustrated in \citep{ST2023}, this approach is intimately related to the seminal work of \cite{B1981} on finite-dimensional stochastic Hamiltonian systems, the theory of which was later studied by \cite{LCO2008}. More specifically to this paper, these methods have also been applied to derive stochastic wave models \citep{crisan2018wave,S2023}. Several studies of the properties of a Stochastic Camassa-Holm (SCH) equation perturbed with transport noise have recently been undertaken. The formulation of the transport noise perturbation is called Stochastic Advection by Lie Transport (SALT) and is based on stochastic variational principles with Lie group symmetry. 

As shown in \cite{crisan2018wave}, wave breaking in the form of peakon emergence occurs with positive probability for the stochastic CH equation with SALT. It was conjectured in \cite{crisan2018wave} that the time-asymptotic solutions of their stochastic CH equation would comprise emergent wave trains of peakons moving along stochastic space-time paths. This was illustrated numerically in the paper \citep{BCH2021}, which presented a finite-element discretisation for this model, and explored the formation of peakons. The simulations presented in \cite{BCH2021} using this discretisation with the SALT approach revealed that peakons can still form in the presence of stochastic perturbations. Peakons can emerge both through wave breaking, as the slope turns vertical, and also without wave breaking as the inflection points of the velocity profile rise to meet at the summit. 

The transport noise form of the stochastic CH equation was also analysed in \cite{albeverio2021stochastic}, whose main result is the proof of the existence and uniqueness of local strong solutions in the Sobolev spaces $H^{2,q}$ for $ 1 < q < \infty$. As mentioned in \cite{albeverio2021stochastic} some results have also been shown in \cite{tang2018pathwise} and \cite{chen2021global} for a \emph{different class} of stochastic perturbations of the CH equation. Namely, the results in this other class are that pathwise solutions of the CH equation exist under perturbation by stochastic linear multiplicative forcing with It{\^o} noise, provided certain initial conditions are sufficiently smooth. However, those studies do not cover the transport noise case, where the diffusion coefficient is generated by an unbounded linear operator. The same may be said about \cite{chen2021global,chen2022stochastic}, although transport noise is mentioned for the stochastic two-component CH equation in an appendix of \cite{chen2023well}. Moreover, \cite{cruzeiro2017stochastic} use an alternative stochastic variational approach to derive a different stochastic Camassa-Holm equation. Finally, \cite{albeverio2021stochastic} re-emphasise the importance of studying the stochastic CH equation because of its geometric and physical motivations, and its relevance in nonlinear wave theory and geophysical applications.

\subsection{Objectives of the paper}
The primary objective of this paper is the comparison of solutions of \emph{stochastic} Camassa--Holm (SCH) equations for two different structure-preserving types of stochasticity. The development of structure-preserving stochastic models of Hamiltonian fluid dynamics may be split into two approaches, commonly referred to in the literature as Stochastic Forcing by Lie Transport (SFLT) and Stochastic Advection by Lie Transport (SALT). Stochastic Option 1 (SFLT) preserves the energy (Hamiltonian), whilst stochasticity perturbing the Poisson operator. Stochastic Option 2 (SALT) preserves the Poisson operator, while it introduces stochasticity into the energy Hamiltonian. The comparison between these two alternatives is based on how introducing either option affects the coherent structures of the deterministic CH equation; namely, its solitary wave solutions known as `peakons' for the peaked shape of their profiles. Comparing the SCH solution behaviour for these two choices of stochasticity reveals the roles of the energy Hamiltonian and the Poisson operator. Said briefly, the metric in the kinetic-energy Hamiltonian is responsible for the shape of the coherent structures in the solitary wave solution (peakons) and the Poisson operator is responsible for creating these coherent structures.

\section{The Model}\label{sec-2}

\subsection{Physical derivation of the deterministic Camassa-Holm equation}\label{sec-2.1}

An asymptotic expansion in the dimension-free scaling parameters $\epsilon_1^2<\epsilon_2<\epsilon_1\ll1$ of the free-surface Euler fluid equation for long wave-length, small-amplitude nonlinear water waves of elevation, $\eta$, yields \citep{dullin2004asymptotically}
\begin{align*}
0 &= \eta_t+\eta_x+\frac{3}{2}\epsilon_1\,\eta\, \eta_x
+\frac{1}{6}\epsilon_2\,\eta_{xxx}
-\frac{3}{8}\epsilon_1^2\,\eta^2\,\eta_x\\
      & 
+\,\epsilon_1\epsilon_2\left(\frac{23}{24}\eta_x\,\eta_{xx}
+ \frac{5}{12}\eta\,\eta_{xxx}\right) + 
\epsilon_2^2 \frac{19}{360}\eta_{xxxxx} + o(\epsilon_1^2,\epsilon_2)
\,.\end{align*}
At this order of asymptotic expansion, certain choices of the three free constant parameters $\alpha_1$, $\alpha_2$, and $\beta$ in the (non-local) Kodama normal form transformation \citep{kodama1997obstacles},   
\begin{align*}
    \eta = u + \epsilon_1\left( \alpha_1 u^2 + \alpha_2 u_x \partial_x^{-1}
u\right) + \epsilon_2 \,\beta u_{xx}
\,,
\end{align*}
result in the deterministic CH equation, expressed as \citep{camassa1993integrable}
\begin{equation}\label{CHe1}
    u_t - \alpha^2u_{txx} = - 3uu_{x} + 2\alpha^2u_{x}u_{xx} + \alpha^2 uu_{xxx} \,,
\end{equation}
or, equivalently,
\begin{equation}\label{CH-mom}
    m_t = - (\partial m + m \partial)u \quad\hbox{where}\quad m := (1-\alpha^2\partial^2)u
\,.\end{equation}

The `\textit{hydrodynamic}' form of the CH equation is \citep{camassa1993integrable}
\begin{align}
\begin{split}
u_{t}+uu_{x} & =-\,\partial_{x}\int_{\mathbb{R}}
\frac{1}{2}\,e^{-|x-y|/\alpha}
\left(
u^{2}(y,t)+\frac{1}{2}u_{y}^{2}(y,t)\right) \,dy
\\&=:-\, K\ast \Big(u^{2}+\frac{1}{2}u_{x}^{2}\Big)_x
\,.\label{GeoSprayCH}
\end{split}
\end{align}
Here, the \emph{peaked} kernal $K(|x-y|) = \tfrac{1}{2}\exp (-|x-y|/\alpha)$ in the convolution integral in \eqref{GeoSprayCH} is the Green function for the Helmholz operator $(1-\alpha^2\partial^2)$ appearing in equation \eqref{CH-mom}. Solutions of equation \eqref{GeoSprayCH} determine the vector field $u= \dot{g}_t {g}_t^{-1}$ which generates the smooth invertible flow $g_t\in {\rm Diff}(S^1)$ governing the solution of the CH equation \eqref{GeoSprayCH}. The geometric content of this statement may be illustrated by writing the CH equation in \eqref{CH-mom} equivalently in terms of the tangent to the push-forward ${g_t}_*$ of the right action of $g_t$ on the 1-form density $m(x,t) dx\otimes dx$. Namely, 
\begin{equation}\label{CH-mom2}
 \frac{d}{dt}{g_t}_*\big(m(x,t) dx\otimes dx\big) = - \,{g_t}_* {\cal L}_{\dot{g}_t {g}_t^{-1}}(m dx\otimes dx)\,.
\end{equation}
This tangent relation also defines the Lie derivative ${\cal L}_u$ along the characteristic curves of the vector field $u= \dot{g}_t {g}_t^{-1}$. Upon evaluating ${g}_t$ in \eqref{CH-mom2} at the identity, ${g}_t|_{t=0}={Id}$, the CH equation may finally be written geometrically as, 
\begin{equation}\label{CH-mom3}
\partial_t \big(m(x,t) \,dx\otimes dx\big) = -\,{\cal L}_{u}\big(m(x,t) \,dx\otimes dx\big)
\,.
\end{equation}

\subsubsection*{Solitary wave solutions}
The CH equation \eqref{CHe1} is known for its solitary wave solutions. 
The deterministic CH equation produces  
trains of peaked solitary wave solutions called `peakons', given by the formula
\begin{equation}  \label{peakontrain-soln-u}
u(x,t)=\frac12\sum_{a=1}^Mp_a(t)\mathrm{e}^{-|x-q_a(t)|/\alpha}  
= \sum_{a=1}^M p_a(t)K(x-q_a(t))
\,.\end{equation}
Since $u=K*m$, the formula \eqref{peakontrain-soln-u} for the peakon solitary wave solution $u(x,t)$ implies the following travelling wave formula for $m(x,t)$, as 
\begin{equation}  \label{peakontrain-soln-m}
m(x,t)= \sum_{a=1}^Mp_a(t)\delta \big(x-q_a(t)\big)  
\,,\end{equation}
in which the $\delta$-function represents \emph{singular} solitary waves. 

Wave trains of peakons emerge in finite time from any smooth confined symmetric initial profile for the velocity.  Hence, these wave trains track the emergence in finite time of the singular solutions appearing in \eqref{peakontrain-soln-m}. The peakon wave trains interact with each other nonlinearly, as shown in Fig. \eqref{Figure1} below.
\begin{figure}[h!]
\begin{center}
\includegraphics[width=0.7\textwidth]{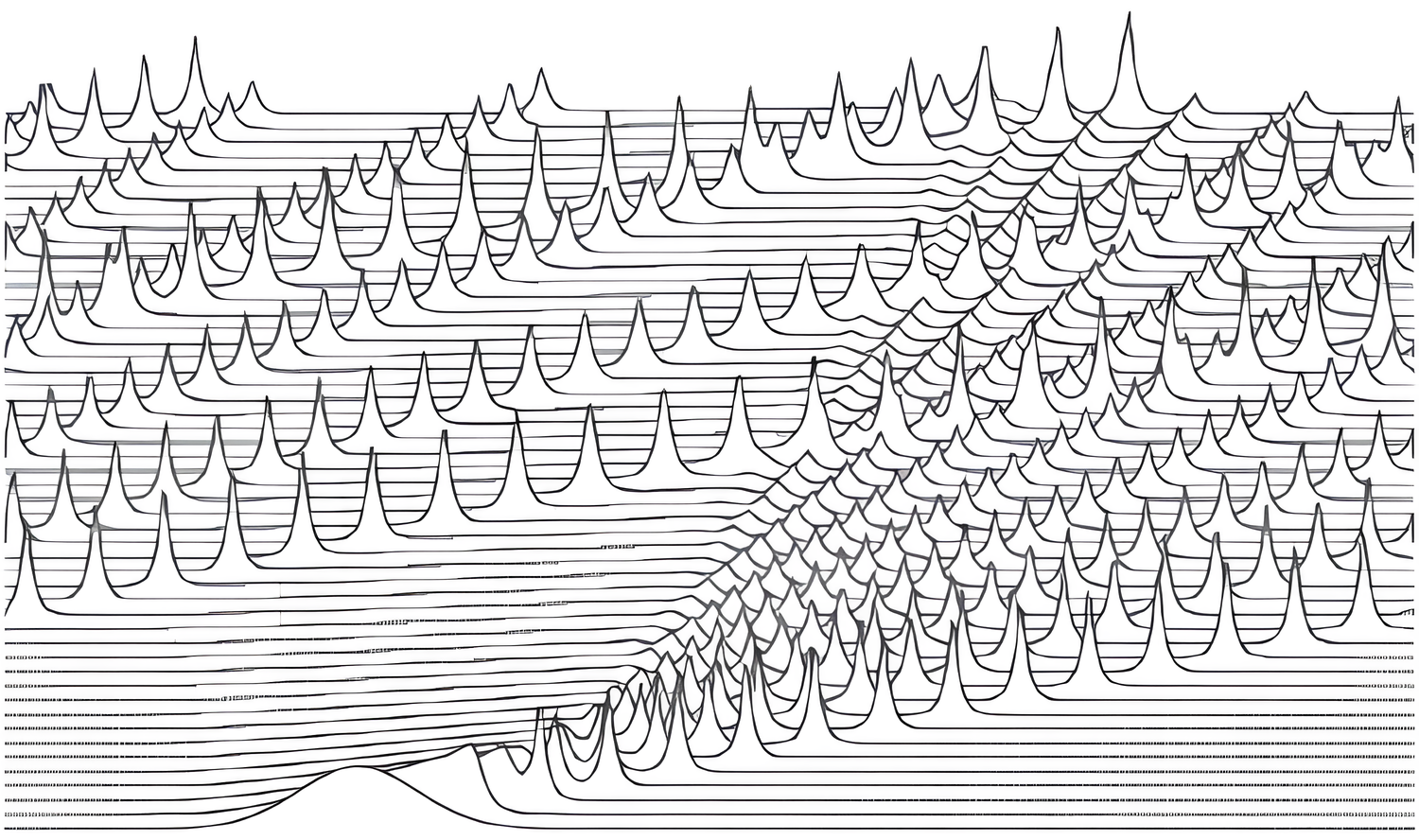}
\end{center}
\caption{\label{Figure1} Peakons emerging from a smooth confined symmetric initial condition.}
\end{figure}
Given the extent of the peakon literature now and the need to move on to our paper's primary objectives, the introduction to the remarkable solution properties of the deterministic CH equation must be brief. For reviews of the deterministic peakon literature, see, e.g., \cite{holm2006peakons,lundmark2022view} and references therein.

\subsection{A Hamiltonian structure for the deterministic Camassa-Holm equation}

To obtain a Hamiltonian structure of the CH equation, we introduce the momentum variable, $m  =u- \alpha^2 u_{xx}$. Note that, formally, the configuration space is the group of diffeomorphisms of the circle ${\rm Diff}(S^1)$, the algebra is then the space of one dimensional vector fields $\mathfrak{X}(S^1)$ and its dual, $\mathfrak{X}^*(S^1)$, containing the momentum, is the space of 1-form densities. Thus, the momentum variable is equipped with a basis, $m dx\otimes dx$. For brevity in notation, the pairing between elements $u\in \mathfrak{X}(S^1)$ and its dual $m\in \mathfrak{X}^*(S^1)$ is written as $\langle m, u \rangle_{\mathfrak{X}^*\times\mathfrak{X}} = \int_{S^1} mu \,dx$.

Expressed in terms of $m$, one of the Hamiltonian structures of the CH equation is given by
\begin{align}\label{eqn:CH_ham}
    m_t=D_m \frac{\delta H}{\delta m} \,,
\end{align}
where the Hamiltonian $H$ and operator $D_m$ are 
\begin{align}
H & =\frac{1}{2} \int_{S^1} u^2+\alpha^2 u_x^2 \,d x=\frac{1}{2} \int_{S^1} u m \,d x 
\,,\\
D_m & =-(m \partial+\partial m)
\,,\quad
\frac{\delta H}{\delta m} = u
\,.
\end{align}
To compute the required derivative of the Hamiltonian, recall that the chain rule for functional derivatives of $F(u,m)$ leads to
\begin{align*}
\frac{\delta F}{\delta m}=\frac{\partial u}{\partial m} \frac{\delta F}{\delta u}=\left(1-\alpha^2 \partial^2\right)^{-1} \frac{\delta F}{\delta u}\,.
\end{align*}

Thus, from the Hamiltonian function $H$ together with the chain rule, we get
\begin{align*}
\frac{\delta H}{\delta m} & =\left(1-\alpha^2 \partial^2\right)^{-1} \frac{\delta H}{\delta u} = \left(1-\alpha^2 \partial^2\right)^{-1} (u - \alpha^2 u_{xx})  =\left(1-\alpha^2\partial^2\right)^{-1}\left(1-\alpha^2\partial^2\right) u=u\,.
\end{align*}
It can then be immediately verified that equation \eqref{eqn:CH_ham} is indeed the CH equation \eqref{CH-mom3} since we have
\begin{align*}
m_t & =-\left(m \partial+\partial m\right) u=-\left(m u_x+(mu)_x\right) \\
& =-3 u u_x+2 \alpha^2 u_x u_{x x}+\alpha^2 u u_{x x x} \,.
\end{align*}

The choice of Hamiltonian structure here has been made since it is the \emph{Lie-Poisson} structure. That is, the Poisson bracket for the system is the Lie-Poisson bracket where the algebra is taken to be the space of one-dimensional vector fields equipped with the commutator. As such, a function of the momentum $f \in C^\infty(\mathfrak{X}^*(\mathbb{R}))$ evolves according to
\begin{equation}\label{eqn:LPB_CH}
\begin{aligned}
    \frac{df}{dt} = \{f,h\} &:= -\left\langle m  , \left[ \frac{\delta f}{\delta m} , \frac{\delta h}{\delta m} \right] \right\rangle_{\mathfrak{X}^*\times\mathfrak{X}} = -\int_{S^1} m \left( \frac{\delta f}{\delta m} \frac{\partial}{\partial x} \frac{\delta h}{\delta m} - \frac{\delta h}{\delta m} \frac{\partial}{\partial x} \frac{\delta f}{\delta m} \right) \,dx
    \\
    &= -\int_{S^1} \frac{\delta f}{\delta m}\left( m\frac{\partial}{\partial x}\frac{\delta h}{\delta m} + \frac{\partial}{\partial x}\left(m\frac{\delta h}{\delta m}\right) \right) \,dx = -\int_{S^1}\frac{\delta f}{\delta m}D_m\frac{\delta h}{\delta m} \,dx \,.
\end{aligned}
\end{equation}
Notice that the above calculation demonstrates that the Hamiltonian operator $D_m$ in fact denotes the Lie derivative, $\mathcal{L}_{\Box}m$, or coadjoint representation, $\ad^*_{\Box}m$, for the action of vector fields on 1-form densities. We have also demonstrated that the Camassa--Holm equation \eqref{eqn:CH_ham} can be written equivalently in Lie--Poisson bracket notation \citep{holm1998euler} as $m_t = \{ m , h \}$. 

\subsection{Derivation of the stochastic Camassa-Holm equations}\label{sec-2.2}

Structure-preserving stochastic perturbations of fluid equations can be made by exploiting this geometric structure. Since this structure is Lie-Poisson, these stochastic perturbations will be analogous to the stochastic advection/forcing by Lie transport (SALT/SFLT) approaches discussed in \citep{H2015} and \citep{HH2021}, respectively.

\subsubsection{SFLT: An energy preserving stochastic perturbation}\label{sec-2.3}

An energy-preserving stochastic perturbation of the CH equation is introduced using the $D_m$ operator, as follows
$$
\begin{aligned}
\dt m & = D_{m\,dt + \sum_k f_k\circ dW^{k}_t} \dfrac{\delta H}{\delta m} := D_m \dfrac{\delta H}{\delta m} d t+ \sum_{k=1}^K D_{{f}_k} \dfrac{\delta H}{\delta m} \circ d W^{k}_t \\
& =-(m \partial+\partial m) u d t- \sum_{k=1}^K\left({f}_k \partial+\partial {f}_k\right) u \circ d W^{k}_t \\
& =\left(-3 u u_x+2 u_x u_{x x}+u u_{x x x}\right) d t- \sum_{k=1}^K \left(2 {f}_k u_x+{{f}_{k}}_{x} u\right) \circ d W^{k}_t \,,
\end{aligned}
$$
where $W^{k}_t$ is a Brownian motion in time and $\circ$ denotes Stratonovich integration. A slight abuse of notation allows us to write
$$
\dt m=-\left(\left(m d t+{f}_k \circ d W^{k}_t \right) u_x+\left(\left(m d t+{f}_k \circ d W^{k}_t\right) u\right)_x\right).
$$
Henceforth, we express the stochastic Camassa-Holm equation (SFLT type) as follows:
\begin{equation}\label{SFLTe1}
 \begin{array}{ll}
 m = u - \alpha^2 \partial_{xx}u, \\
\dt m + \left(\widetilde{m} u_x + (\widetilde{m} u )_x \right ) = 0,
\end{array}  
\end{equation}
where the perturbed momentum due to SFLT noise is 
\[
\widetilde{m} = m  \dt t + \sum_{k=1}^{K} {f}^{k}(x) \circ dW^{k}_t \,.
\]
Details of the spatially modulated amplitudes ${f_k}$ will be declared in the numerical section. The model \eqref{SFLTe1} describes the evolution of velocity $u$ in time $t \in [0,T)$ on the spatial domain $[0,L_d]$ with periodic boundary conditions. 

Next, we will prove that the resulting stochastic CH equation is energy-preserving. Define the energy functional as 
\[
E = H = \int_{0}^{L_d} \left(\dfrac{1}{2}u^2 + \dfrac{\alpha^2}{2} u^{2}_{x}\right)  \dt x \,.
\]
\begin{proposition}\label{Ctsenergythm}
    The energy, $E$, is preserved by the flow of \eqref{SFLTe1}.
\end{proposition}
\begin{proof}
    Following equation \eqref{eqn:LPB_CH}, the Lie-Poisson bracket can be written in terms of the operator $D_{\Box}$ and the energy, as a function of momentum, evolves according to
    \begin{align*}
        dE &= -\int_{S^1}\frac{\delta f}{\delta m}D_m\frac{\delta h}{\delta m} \,dx\,dt - \sum_{k=1}^K\int_{S^1}\frac{\delta f}{\delta m}D_{{f}_k}\frac{\delta h}{\delta m} \,dx\circ dW^{k}_t
        \\
        &=: \{ E , H \} \,dt + \sum_{k=1}^K \{ E , K \}_{{f}_k}\circ dW^{k}_t  \,,
    \end{align*}
    where $\{ \cdot , \cdot \}_{{f}_k}$ is the `frozen' Lie-Poisson at ${f}_k$. Since these are both Poisson brackets and $E=H$, the result follows. 
\end{proof}

\subsubsection{SALT: A `Poisson structure preserving' stochastic perturbation}

The alternative option is to incorporate stochasticity into the Hamiltonian, rather than in the Hamiltonian operator. That is, we introduce a Hamiltonian for the stochastic component of the flow
\begin{equation*}
    H_k = \int_{S^1}m \xi_k\,dx \,,
\end{equation*}
and consider the equation
$$
\begin{aligned}
\dt m & = D_m \left( u\,dt + \sum_{k=1}^K \xi_k\circ dW^{k}_t \right) := D_m \dfrac{\delta H}{\delta m} d t+ \sum_{k=1}^K D_m \dfrac{\delta H_k}{\delta m} \circ d W^{k}_t \\
& =-(m \partial+\partial m) u d t- \sum_{k=1}^K\left(m \partial+\partial m\right) \xi_k \circ d W^{k}_t \\
& =\left(-3 u u_x+2 u_x u_{x x}+u u_{x x x}\right) d t- \sum_{k=1}^K \left(2 m{\xi_k}_x + m_x\xi_k\right) \circ d W^{k}_t \,.
\end{aligned}
$$
This approach to stochastically perturbing the CH equation, studied by \cite{BCH2021}, is analogous to the `SALT' approach to 
Casimir-preserving stochastic perturbations of physical fluid models \citep{H2015}. The choice of the Hamiltonians $H_k$ reflects the fact that this approach was first developed on the Lagrangian framework, and corresponds to a perturbation of material transport. The `SALT' approach ensures that the stochastic model preserves invariants which arise as a result of being Casimirs of the Poisson bracket and, moreover, ensures that the dynamics remains on the coadjoint orbit of the Lie co-algebra. As such, it maintains certain properties of the flow which correspond to the bracket structure. For a full discussion of this, see \cite{ST2023}.

As with the energy-preserving stochastic model \eqref{SFLTe1}, we will express the stochastic Camassa-Holm equation (SALT type) in a simplified notation as
\begin{equation}\label{SALTe1}
 \begin{array}{ll}
 m = u - \alpha^2 \partial_{xx}u, \\
\dt m + \left(m \tilde{u}_x + (m \tilde{u} )_x \right ) = 0,
\end{array}  
\end{equation}
where the perturbed velocity due to SALT noise is 
\[
\tilde{u} = u  \dt t + \sum_{k=1}^{K} \xi^{k}(x) \circ dW^{k}_t \,,
\]
where, as for the previous equation, the amplitudes $\xi_k$ will be taken to be sine/cosine functions, declared in the numerical section.

\section{Numerical Approximation}\label{sec-2.4}
In this section, we will establish numerical schemes of the SCH equations of both type \eqref{SFLTe1} and \eqref{SALTe1}. Also, we briefly mention the energy-preserving property of the numerical method.

\subsection{SFLT model}\label{sec-2.3.1}
To study the numerical simulation, we follow the method of lines approach. Firstly, we discretise continuous models in the spatial variables followed by the temporal variable. 
For spatial discretisation, we will consider a  continuous Galerkin finite element technique on a uniform mesh of
the interval $I=[0,L_d]$ with $N$ cells of width $h=L_d /N$. The respective finite element space, \emph{i.e.}, the continuous Galerkin space $CG_1$ on the mesh is denoted as $V_h$. One can choose a higher-order spatial discretisation, but while keeping in mind the convergence of stochastic PDEs, we have restricted ourself to the former case. The semidiscrete numerical scheme seeks
$m_{h}(t)\in V_h$ and $u_{h}(t)\in V_h$ such that
\begin{equation}\label{SFLTde1}
\begin{array}{ll}
	(u_h ,\psi_h) + \alpha^2 ( \partial_{x}u_h,\partial_{x}\psi_h) - (m_h,\psi_h) =0 \,, \quad \forall \psi_h \in V_{h} \,,\\[4pt]
	 (\dt m_h, \phi_h) + (\widetilde{m}_h \partial_{x}u_h, \phi_h) - (\widetilde{m}_h u_h, \partial_{x}\phi_h)=0 \,, \quad \forall \phi_h \in V_{h} \,,
\end{array}
\end{equation}
where the perturbed momentum is $\widetilde{m}_h$ is defined by replacing  $m$ with $m_h$ and is given by
\[
\widetilde{m}_h = \left(m_h  \dt t + \sum_{k=1}^{K} {f}^{k}\circ dW^{k}_t \right) \,.
\]

For the time discretisation, we select a uniform time step $\Delta t = T/M$ and $t_{n} = n \Delta t, \,\,n=1,2,\ldots, N_{T}$, and solve for 
$m^n\approx m(x,t_n)$ and $u^n\approx u(x, t_n)$. Note that, 
$\Delta W^{n,k}$ is a normally distributed  with  mean zero and variance $\Delta t$. Then we use $\Delta W^{n,k}$ in an implicit midpoint rule discretisation (leading to a Stratonovich method since $\Delta W^{n,k}$ is multiplied by $u^{n+1/2}_h$, not $u^n_h$).

For the fully discrete scheme, one must find $m^{n}_h, u^n_h \in V_h$ which solves 
\begin{equation}\label{SFLTde2}
\begin{array}{ll}
	(u^{n+1}_h,\psi_h) + \alpha^2 ( \partial_{x}u^{n+1}_h,\partial_{x}\psi_h) - (m^{n+1}_h,\psi_h) =0 \,, \quad \forall \psi_h \in V_{h} \,,\\[4pt]
	 (m^{n+1}_h -m^{n}_h , \phi_h)  
+  (\widetilde{m}^{n+1/2}_{h}\partial_{x}u^{n+1/2}_h, \phi_h)  -  (\widetilde{m}^{n+1/2}_h u^{n+1/2}_h , \partial_{x}\phi_h)=0 \,, \quad \forall \phi_h \in V_{h} \,,
\end{array}
\end{equation}
where $\widetilde{m}^{n+1/2}_h$ is given as
\[
\widetilde{m}^{n+1/2}_h = \left(m^{n+1/2}_h  \Delta t +  \sum_{k=1}^{K} {f}^{k}_h \Delta W^{n, k} \right) \,,\quad\hbox{with}\quad m^{n+1/2}_h = \frac{1}{2}\left(m^{n+1}_h + m^{n}_h \right) \,. 
\]
Note that, if the system is deterministic, there are no ${f}^{k}$ functions that are non-zero, and the equations reduce so that $\widetilde{m}^{n+1/2}_h = m^{n+1/2}_h  \Delta t$.

Similar to the deterministic case, we show that the numerical scheme \eqref{SFLTde2} is energy-preserving. Let us define the discrete energy functional $ E^{n+1}$ by 
\begin{equation}\label{Denergy}
    E^{n+1} = \dfrac{1}{2} \int_{0}^{L_d} \left((u^{n+1}_h)^2 + \alpha^2 (\partial_{x}u^{n+1}_h)^2\right) \dt x \,.
\end{equation}

\begin{proposition}\label{Disenergythm}
For $E$ defined in equation \eqref{Denergy}, we have that $\dfrac{1}{\Delta t}(E^{n+1}-E^{n})=0$ where $u$ evolves according to the scheme \eqref{SFLTde2}.
\end{proposition}
\begin{proof} From the discrete scheme \eqref{SFLTde2}, we have
\begin{align*}
    \dfrac{1}{\Delta t}(E^{n+1}-E^{n}) & = \int_{0}^{L_d}  \left( \dfrac{\big(u^{n+1}_h - u^{n}_h\big)}{\Delta t} u^{n+1/2}_h + \alpha^2 \dfrac{\big(\partial_{x}u^{n+1}_h - \partial_{x}u^{n}_h\big)}{\Delta t} \partial_{x} u^{n+1/2}_h \right) \dt x \\[12pt]
    & = \int_{0}^{L_d} u^{n+1/2}_h \dfrac{\big(m^{n+1}_h - m^{n}_h\big)}{\Delta t} \dt x \\[12pt]
    & = \int_{0}^{L_d} \left(- u^{n+1/2}_h \widetilde{m}^{n+1/2}_{h}\partial_{x}u^{n+1/2}_h +  \widetilde{m}^{n+1/2}_h u^{n+1/2}_h  \partial_{x}u^{n+1/2}_h  \right)\dt x = 0 \,. 
\end{align*}
By using integration by parts and the form of the discrete scheme \eqref{SFLTde2}, we find the last two equalities.
\end{proof}

\subsection{SALT model}\label{sec-2.3.2}
For the discretisation of the SALT model \eqref{SALTe1}, we follow the numerical implementation discussed in \cite{cottercrisansingh_STUOD}. For the sake of completeness, we shall introduce this numerical scheme briefly. Find $m^{n}_h , u^n_h \in V_h$ such that
\begin{equation}\label{SALTde1}
\begin{array}{ll}
	(u^{n+1}_h,\psi_h) + \alpha^2 ( \partial_{x}u^{n+1}_h,\partial_{x}\psi_h) - (m^{n+1}_h,\psi_h) =0 \,, \quad \forall \psi_h \in V_{h} \,,\\[4pt]
	 (m^{n+1}_h-m^{n}_h, \phi_h)  + (m^{n+1/2}_h\partial_{x}\tilde{u}^{n+1/2}_h, \phi_h) - (m^{n+1/2}_h \tilde{u}^{n+1/2}_h \,, \partial_{x}\phi_h)=0 \,, \quad \forall \phi_h \in V_{h}\,,
\end{array}
\end{equation}
where the perturbed term $\tilde{u}^{n+1/2}_h$ is given by
\[
\tilde{u}^{n+1/2}_h = \left(u^{n+1/2}_h \Delta t + \sum_{k=1}^{K} \xi^{k}_h \Delta W^{n, k}\right).
\]
Before concluding this section, we note that the discrete scheme \eqref{SALTde1} is not an energy-preserving integrator, a characteristic consistent with the continuous model \eqref{SALTe1} as previously discussed.

\section{Numerical experiments}\label{sec-3}
In this section, we will study the numerical approximation of equation \eqref{SFLTe1} using the discrete scheme  \eqref{SFLTde2}. In the numerical discussion, we will investigate the formation of peakons using different initial conditions. We will study the numerical simulation  of the newly introduced SFLT model \eqref{SFLTe1} and compare with corresponding simulations of the SALT model \eqref{SALTe1}. We implemented our numerical experiments using the Firedrake software \citep{FiredrakeUserManual}.

In the numerical simulations, we have used $50$ ensemble members to study the behaviour of the stochastic models \eqref{SFLTde1} and \eqref{SALTde1}. In the experiments, the mesh parameters are $N=5000$ and $\Delta t = 0.0025$ and the model parameter $\alpha =1$. For all numerical discussions, the spatial amplitude functions $f_k = \xi_k$ are given by:
\[
f_k(x) =  \xi_k = 
\begin{cases}
\sigma_k \sin\left( \frac{(2(k +1)+ L_d/4) \pi x}{L_d} \right), & \text{if } k \text{ is odd} \\[4pt]
\sigma_k \cos\left( \frac{(2(k +1)+ L_d/4) \pi x}{L_d} \right), & \text{if } k \text{ is even}
\end{cases}
\]
where $\sigma_k \in \{0.01, 0.02, 0.05, 0.1, 0.2\}$.

\begin{example}\label{example1}
In the first example, both stochastic models \eqref{SFLTde1} and \eqref{SALTde1} are solved numerically, starting with the same Gaussian initial velocity profile:
\[
u_0 = \dfrac{1}{2}\exp\left(\dfrac{(x-10)^2}{3^2}\right).
\]
Fig. \ref{fig:SALT_vs_SFLT_ex1} compares the SFLT and SALT ensembles at different times for this example.
Apparently, the SFLT ensemble is more representative as a stochastic spread of uncertainty in the deterministic dynamics than the SALT ensemble. The time sequence of Fig. \ref{fig:SALT_vs_SFLT_ex1} demonstrates that the lack of energy conservation in the SALT case causes the peakons to gain or lose amplitude and speed of propagation. Consequently, as shown in Fig. \ref{fig:peakSALT_vs_SFLT}, the heights of the maximum peaks in the SFLT and SALT ensembles are not identical and propagate at different speeds.
\end{example}

\begin{example}\label{example2}
In the second example, the initial condition for the velocity profile is taken as the antisymmetric peakon-antipeakon configuration,
\[
u_0(x) = \exp\left(-\dfrac{\big|x-\frac{L_d}{4}\big|}{\alpha}\right) - \exp\left(-\dfrac{\big|x-\frac{3L_d}{4}\big|}{\alpha}\right) \,.
\]
Fig. \ref{fig:SALT_vs_SFLT_ex2} compares the time sequences of the SFLT and SALT stochastic ensembles for this example of the head-on collisions produced by the peakon-antipeakon initial condition. As in the first example, the time sequences of Fig. \ref{fig:SALT_vs_SFLT_ex2} show a difference arising because the peakons in the SALT ensemble are changing their amplitude and propagation speed more dramatically than in the SFLT ensemble.
\end{example}


\begin{figure}[!htb]
    \centering

    \begin{minipage}{0.48\textwidth}
        \centering
        \includegraphics[width=.85\linewidth]{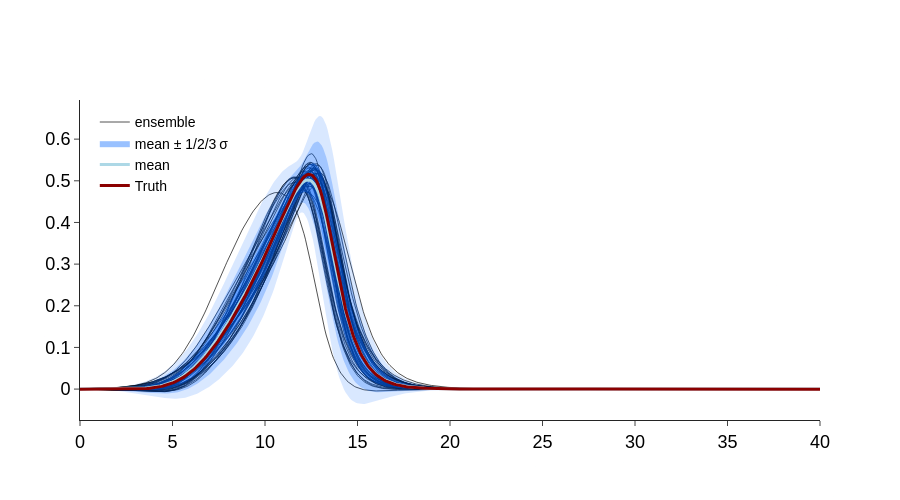}
        \subcaption*{\footnotesize SFLT: $t=2.5$}
    \end{minipage}
    \hspace{-0.95em}
    \begin{minipage}{0.48\textwidth}
        \centering
        \includegraphics[width=.85\linewidth]{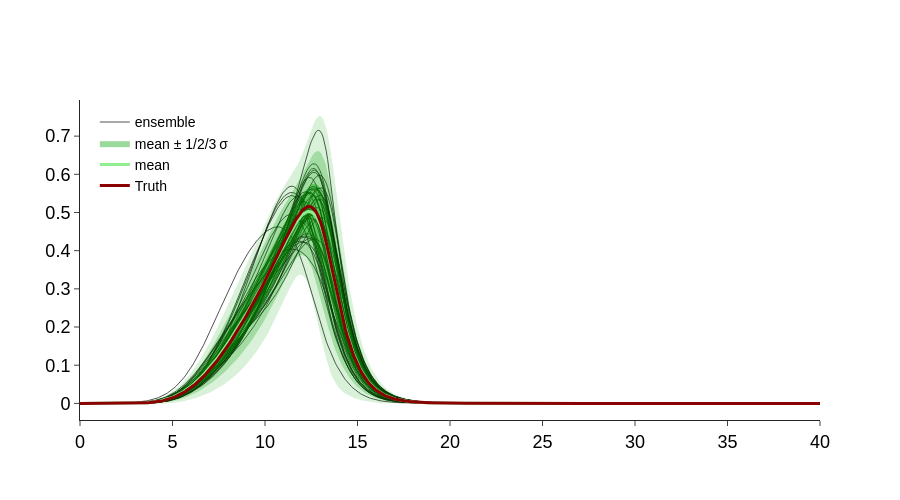}
        \subcaption*{\footnotesize SALT: $t=2.5$}
    \end{minipage}

    \vspace{-0.3em}

    \begin{minipage}{0.48\textwidth}
        \centering
        \includegraphics[width=.85\linewidth]{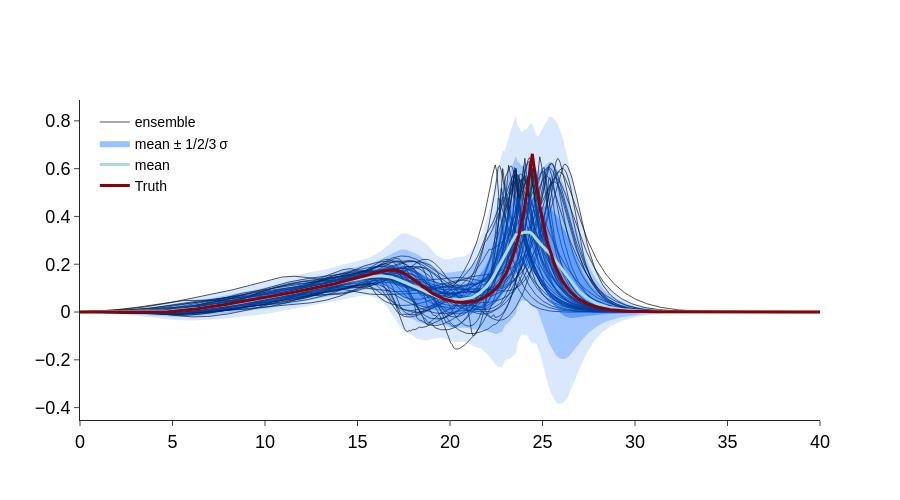}
        \subcaption*{\footnotesize SFLT: $t=20$}
    \end{minipage}
    \hspace{-0.95em}
    \begin{minipage}{0.48\textwidth}
        \centering
        \includegraphics[width=.85\linewidth]{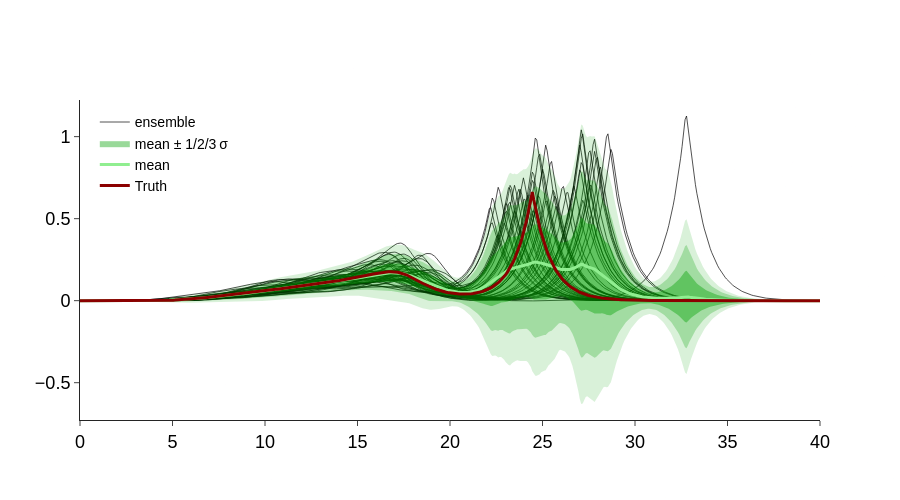}
        \subcaption*{\footnotesize SALT: $t=20$}
    \end{minipage}

    \vspace{-0.35em}

    \begin{minipage}{0.48\textwidth}
        \centering
        \includegraphics[width=.85\linewidth]{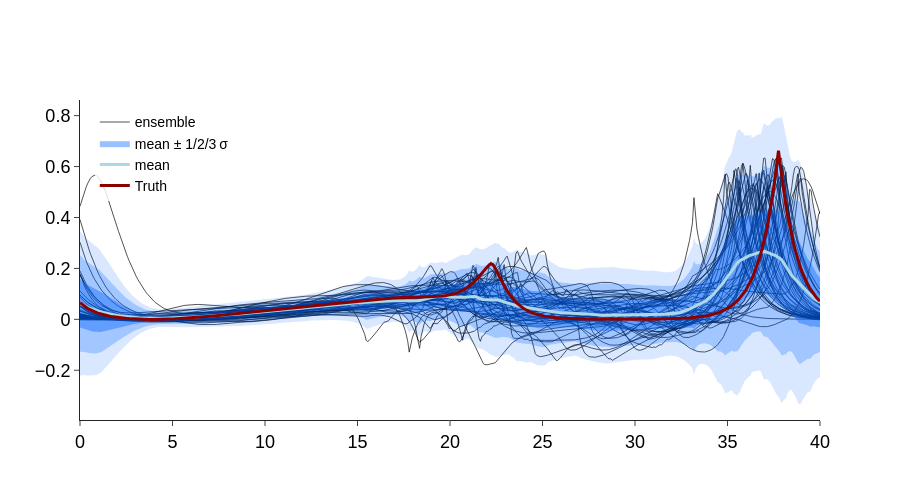}
        \subcaption*{\footnotesize SFLT: $t=40$}
    \end{minipage}
    \hspace{-0.95em}
    \begin{minipage}{0.48\textwidth}
        \centering
        \includegraphics[width=.85\linewidth]{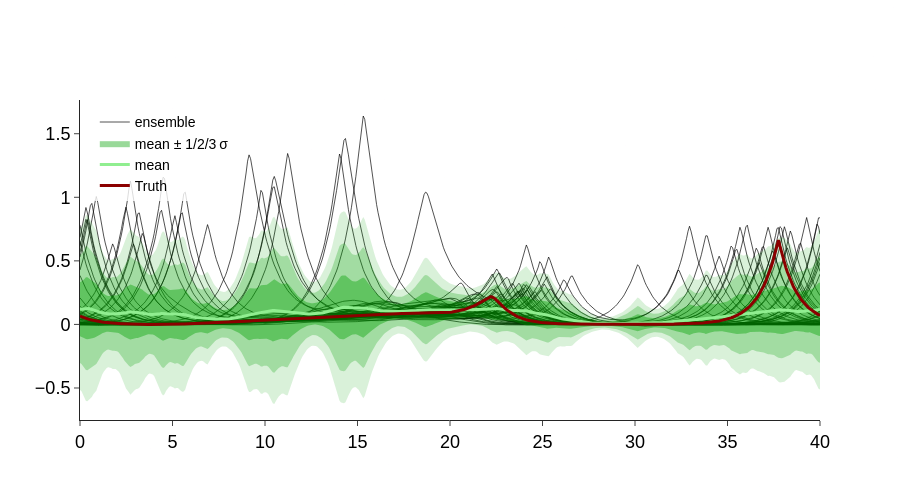}
        \subcaption*{\footnotesize SALT: $t=40$}
    \end{minipage}

    \caption{This figure compares the SFLT and SALT ensembles at different times. Left column: SFLT; Right column: SALT.}
    \label{fig:SALT_vs_SFLT_ex1}
\end{figure}


\begin{figure}[!htb]
    \centering
    \begin{minipage}{0.48\textwidth}
        \centering
        \includegraphics[width=0.85\linewidth]{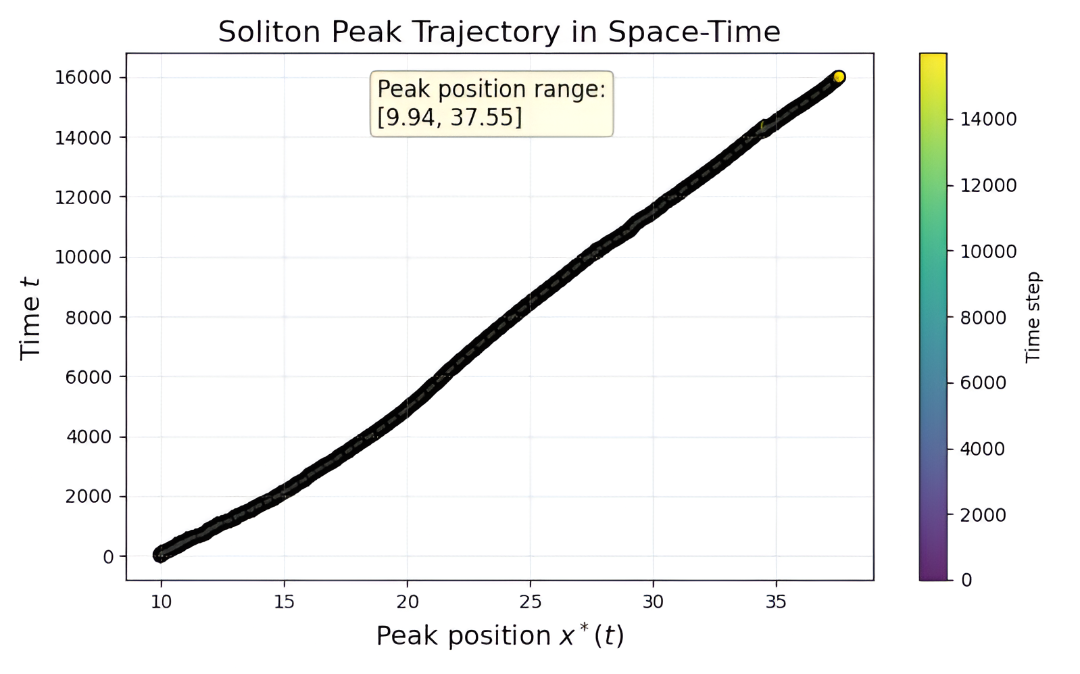}
        \subcaption*{\footnotesize SFLT}
    \end{minipage}
    \hfill
    \begin{minipage}{0.48\textwidth}
        \centering
        \includegraphics[width=0.85\linewidth]{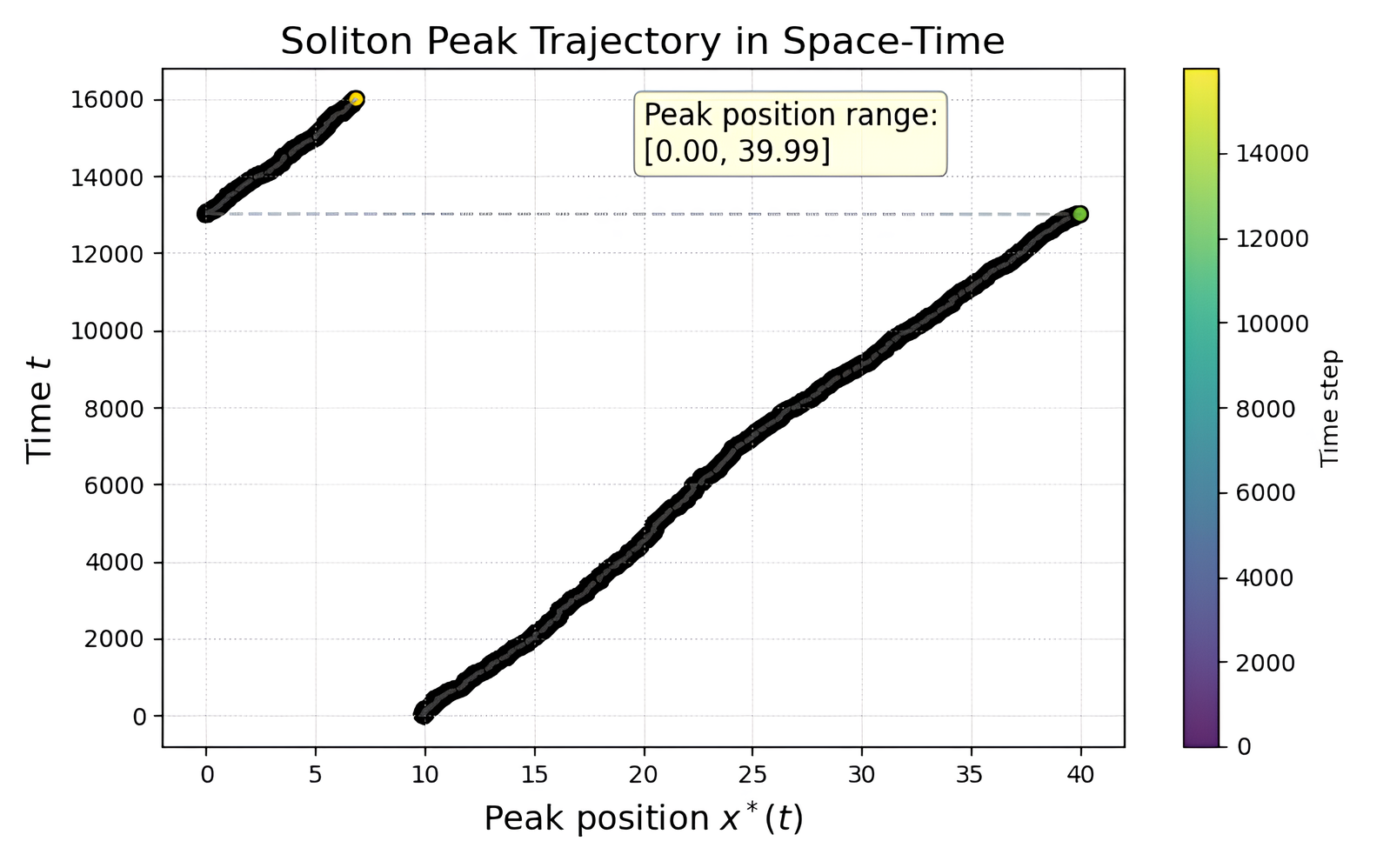}
        \subcaption*{\footnotesize SALT}
    \end{minipage}

    \caption{This figure shows that the heights of the maximum peaks in SFLT and SALT ensembles dynamics propagate at different rightward speeds.}
    \label{fig:peakSALT_vs_SFLT}
\end{figure}

\begin{figure}[!htb]
    \centering
    \vspace{-0.65em}
\begin{minipage}{0.48\textwidth}
        \centering
        \includegraphics[width=.8\linewidth]{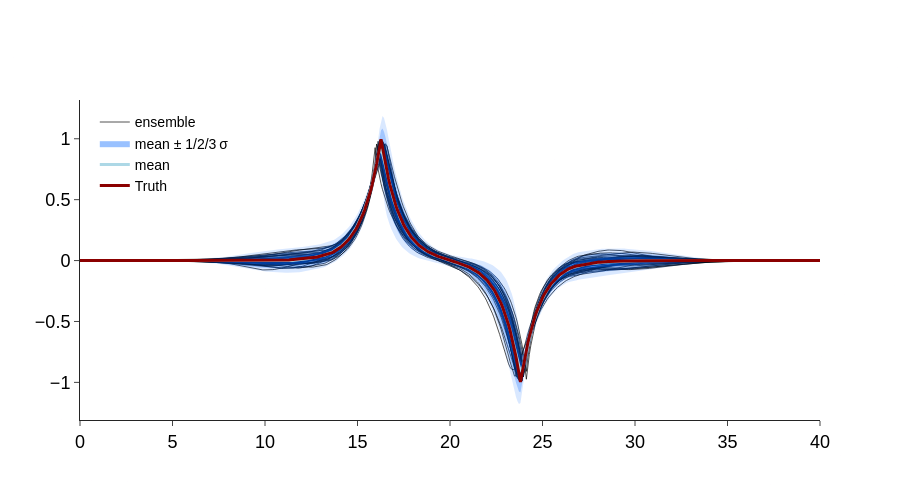}
        \subcaption*{\footnotesize SFLT: $t=2.5$}
    \end{minipage}
    \hspace{-0.95em}
    \begin{minipage}{0.48\textwidth}
        \centering
        \includegraphics[width=.8\linewidth]{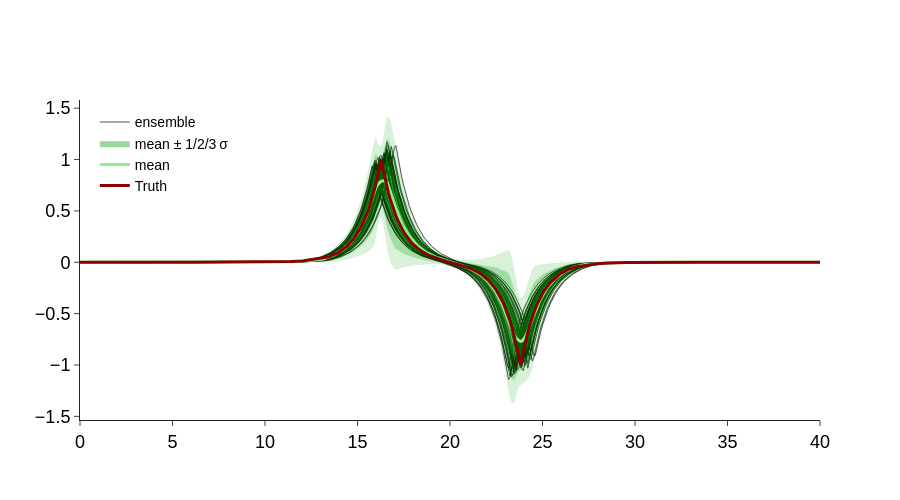}
        \subcaption*{\footnotesize SALT: $t=2.5$}
    \end{minipage}

\vspace{-0.4em}

  \begin{minipage}{0.48\textwidth}
        \centering
        \includegraphics[width=.8\linewidth]{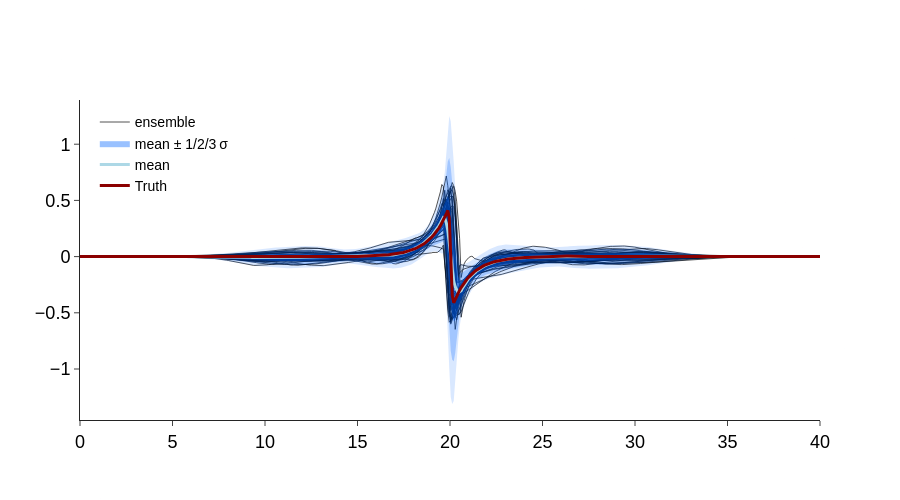}
        \subcaption*{\footnotesize SFLT: $t=10.25$}
    \end{minipage}
    \hspace{-0.95em}
    \begin{minipage}{0.48\textwidth}
        \centering
        \includegraphics[width=.8\linewidth]{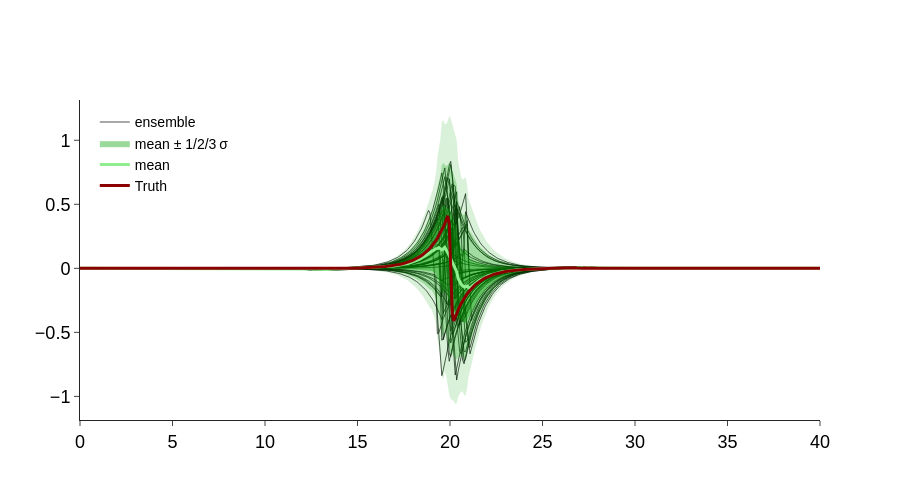}
        \subcaption*{\footnotesize SALT: $t=10.25$}
    \end{minipage}
\vspace{-0.4em}

 \begin{minipage}{0.48\textwidth}
        \centering
        \includegraphics[width=.8\linewidth]{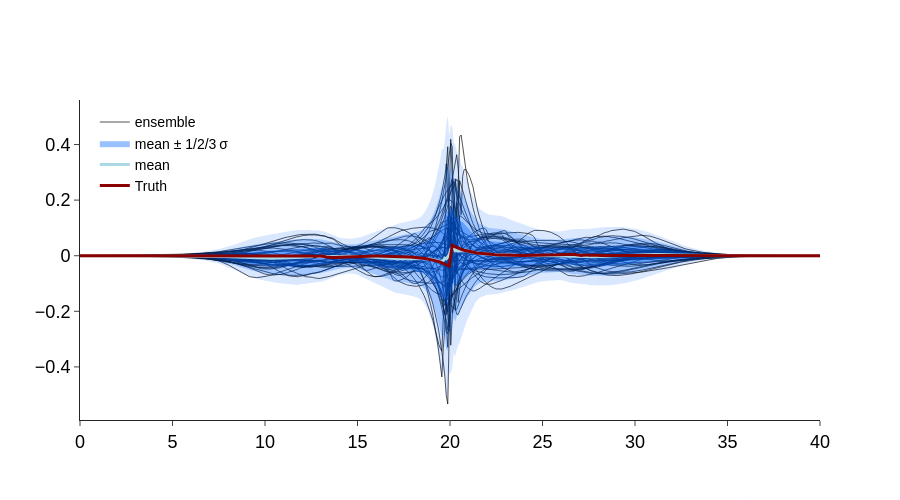}
        \subcaption*{\footnotesize SFLT: $t=10.75$}
    \end{minipage}
    \hspace{-0.95em}
    \begin{minipage}{0.48\textwidth}
        \centering
        \includegraphics[width=.8\linewidth]{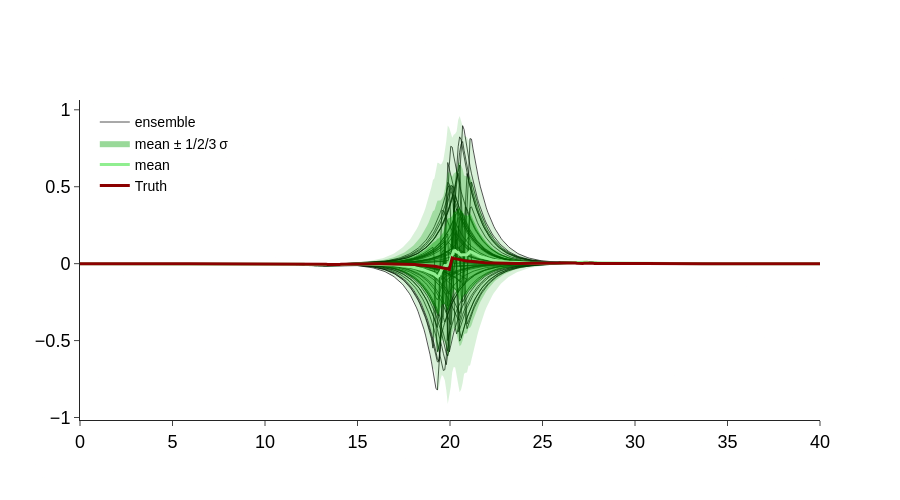}
        \subcaption*{\footnotesize SALT: $t=10.75$}
    \end{minipage}

 \vspace{-0.4em}

\begin{minipage}{0.48\textwidth}
        \centering
        \includegraphics[width=.8\linewidth]{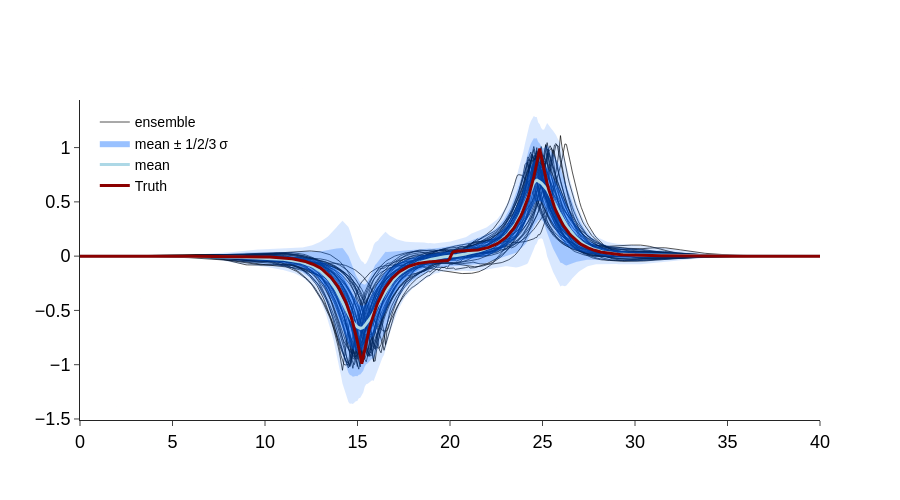}
        \subcaption*{\footnotesize SFLT: $t=16.25$}
    \end{minipage}
    \hspace{-0.95em}
    \begin{minipage}{0.48\textwidth}
        \centering
        \includegraphics[width=.8\linewidth]{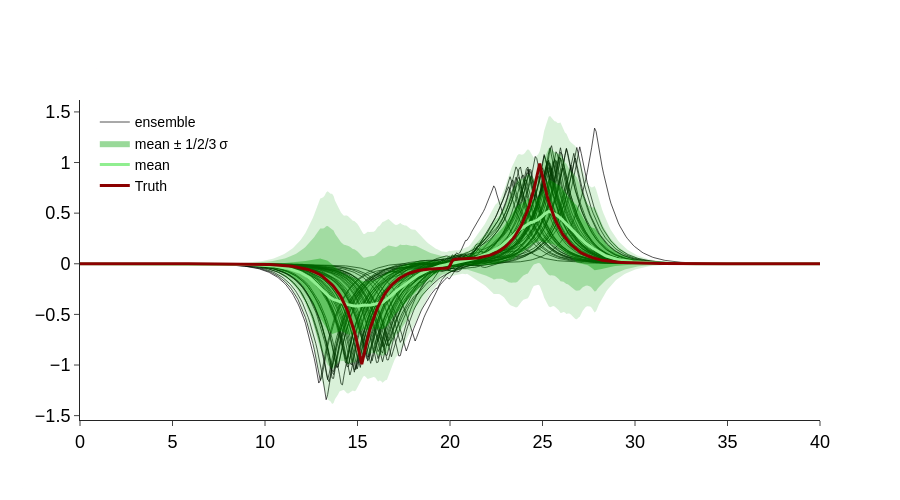}
        \subcaption*{\footnotesize SALT: $t=16.25$}
    \end{minipage}

\vspace{-0.35em}
  \begin{minipage}{0.48\textwidth}
        \centering
        \includegraphics[width=.8\linewidth]{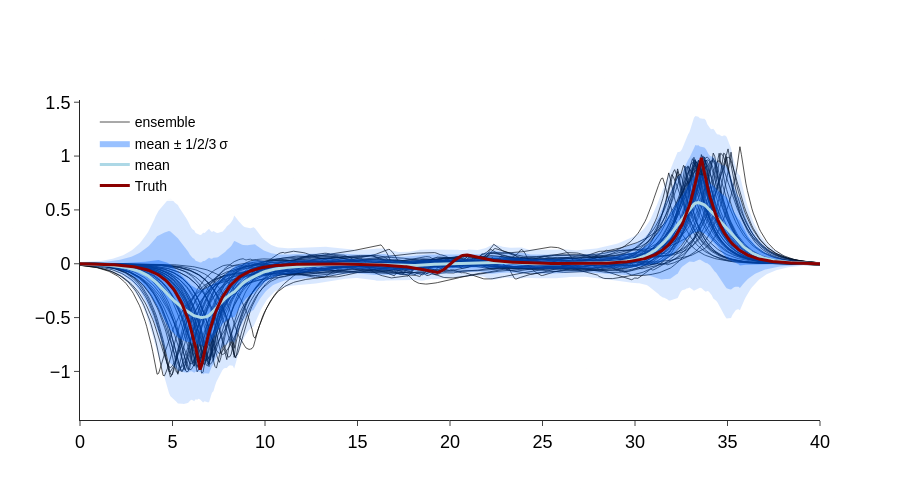}
        \subcaption*{\footnotesize SFLT: $t=25$}
    \end{minipage}
    \hspace{-0.95em}
    \begin{minipage}{0.48\textwidth}
        \centering
        \includegraphics[width=.8\linewidth]{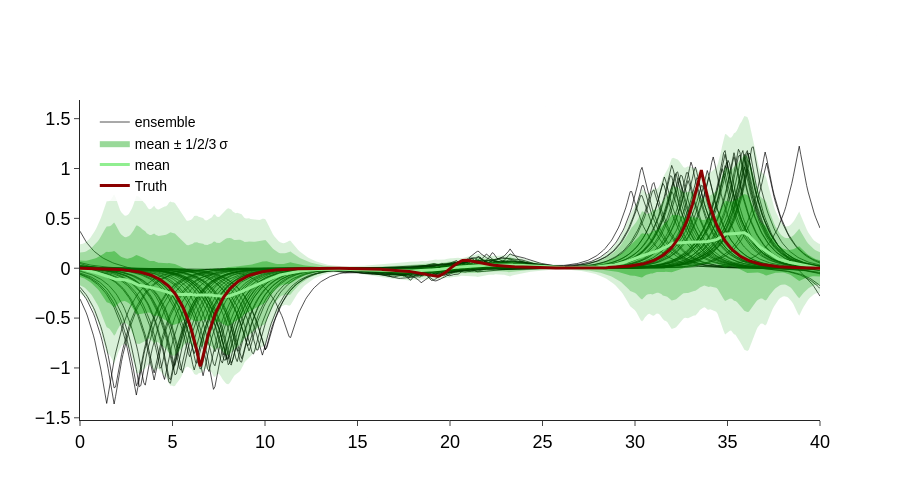}
        \subcaption*{\footnotesize SALT: $t=25$}
    \end{minipage}

    \caption{Comparison of SFLT and SALT ensembles at different times. Left column: SFLT; Right column: SALT.}
    \label{fig:SALT_vs_SFLT_ex2}
\end{figure}

Fig. \ref{energy_ex1}  and Fig. \ref{energy_ex2} show that in both of these examples the total energy is preserved by the SFLT model and is not preserved by the SALT model. To make this comparison in energy, we have quantified the energy of the stochastic models \eqref{SFLTde1} and \eqref{SALTde1} numerically by decomposing the total energy as $E^n = E^n_1 + E^n_2$, with terms $E^n_1 = \frac{1}{2}\|u\|^2_{L^2}$ and  $E^n_2 = \frac{1}{2}\|u_x\|^2_{L^2}$.  


\begin{figure}[!htb]
    \centering
   \begin{minipage}{0.47\textwidth}
        \centering
        \includegraphics[width=0.75\linewidth]{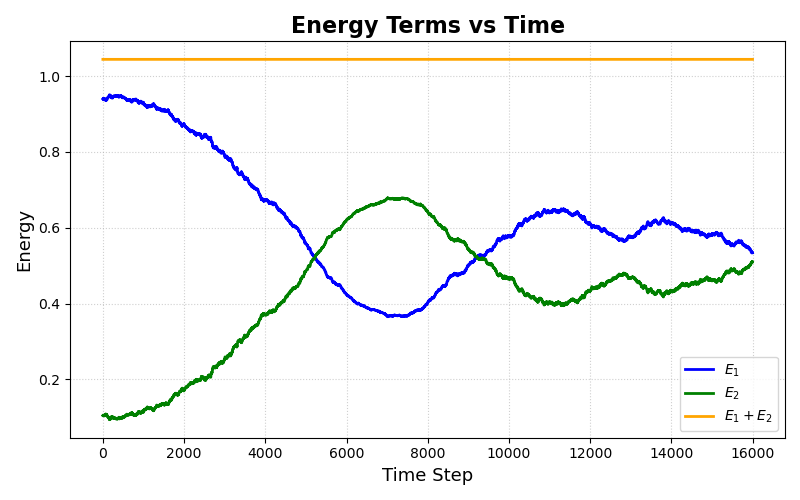}
        \subcaption{\footnotesize SFLT: Energy measure }
    \end{minipage}
    \hfill
    \begin{minipage}{0.47\textwidth}
        \centering
        \includegraphics[width=0.75\linewidth]{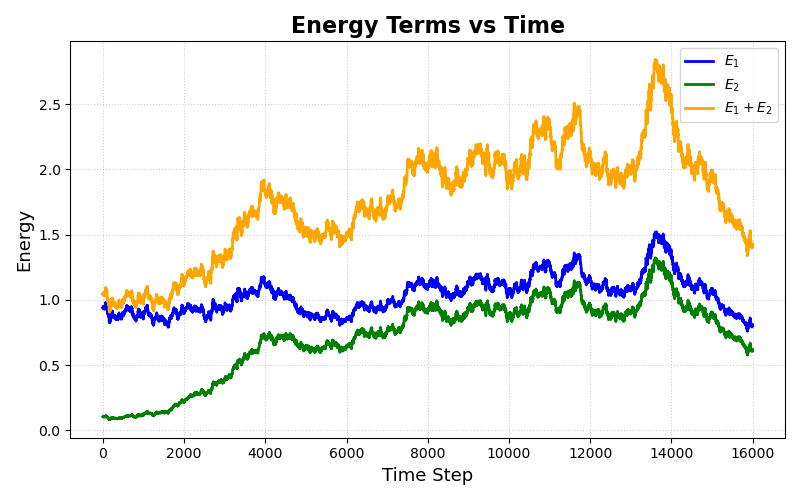}
        \subcaption{\footnotesize SALT: Energy measure}
    \end{minipage}
    \caption{Comparison of SFLT v SALT energy preservation for Example \eqref{example1} }
    \label{energy_ex1}
\end{figure}

\begin{figure}[!htb]
    \centering
   \begin{minipage}{0.47\textwidth}
        \centering
        \includegraphics[width=0.75\linewidth]{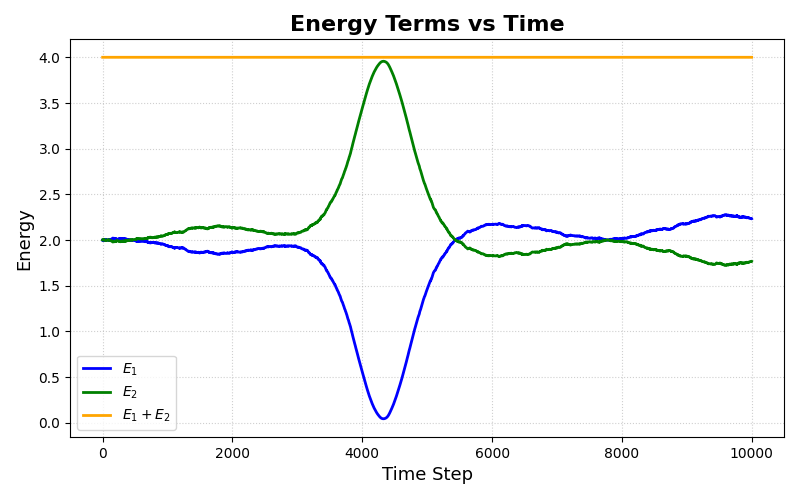}
        \subcaption{\footnotesize SFLT: Energy measure }
    \end{minipage}
    \hfill
    \begin{minipage}{0.47\textwidth}
        \centering
        \includegraphics[width=0.75\linewidth]{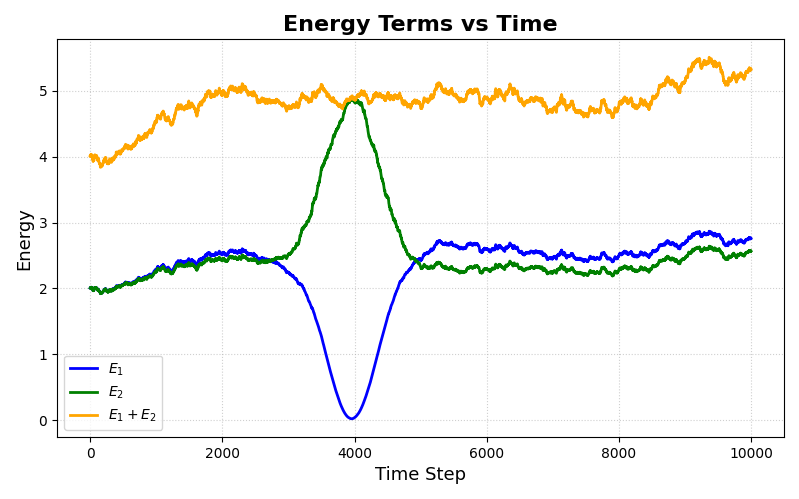}
        \subcaption{\footnotesize SALT: Energy measure}
    \end{minipage}
    \caption{Comparison of of SFLT v SALT energy preservation for Example \eqref{example2} }
    \label{energy_ex2}
\end{figure}


\section{Conclusion and Outlook}\label{sec-4}
The paper has compared two alternatives approaches for introducing stochasticity into the  deterministic CH equation, in terms of their effects on its solitary wave solutions known as `peakons'. Comparing the solution behaviour for these two choices of stochasticity has revealed the roles of the energy Hamiltonian and the Poisson operator in peakon dynamics. As mentioned in the introduction and illustrated in the article, the metric in the kinetic-energy Hamiltonian is responsible for the shape of the coherent structures in the solitary wave solution, which was also verified in previous numerical simulations \citep{fringer2001integrable}. In contrast, the Poisson operator is responsible for the \emph{creation} of these coherent structures. It can be conjectured that the difference between the conserved properties of the model in each case presented here has some effect on the emergence of peakons from smooth initial data. The fact that stochasticity does not inhibit this creation was verified theoretically in \cite{crisan2018wave} and numerically in \cite{BCH2021}. 

In this article, we see that both the SALT and SFLT cases exhibit this solution behaviour. The key difference found in this work is that the SALT perturbation of the transport velocity more dramatically influences the amplitude and speed of the peakon propagation, thus generating a wider spread in the ensemble of solutions than the energy-conserving SFLT perturbation of the momentum density. This illustrates the notion that the energy is responsible for the shape of the coherent structure and that the energy conservation of the SFLT model limits its ability to influence this feature of the dynamics. A more specific investigation about the effect of the two approaches on the \emph{generation} of peakons is a topic of ongoing research.

Fig. \ref{Figure1} shows that the choice of the peakon width $(\alpha)$ determines the physical size of a given peakon and thereby determines its effective range of interaction. This dynamical dependence on the peakon widths suggests considering an ensemble of \emph{stochastic interactions among peakons of different widths}. This type of multiscale peakon interaction has recently been introduced for the deterministic case in \cite{holm2025multiscalegeodesicflows}, inspired by a well-known rhyme attributed to L.F. Richardson's observations on turbulence \cite{richardson1922weather}. The current work on stochastic peakon dynamics suggests that an investigation of stochastic multiscale peakon interactions may also be fruitful.

Another challenging extension of the present work that lies beyond the present scope would be to make a comparison similar to the present one between SFLT and SALT for the Camassa--Holm equation in higher dimensions, also known as the Euler-alpha model  \cite{chen1998camassa,foias2001navier,foias2002three}. In particular, the issue of well-posedness appears to be surmountable for both the SFLT and SALT stochastic models of the Camassa--Holm equation in higher dimensions.

\subsection*{Acknowledgements} 
We are grateful to our friends and colleagues T. Bendall,  C. J. Cotter, D. Crisan, R. Hu, S. Takao, and J. Woodfield, for their interest, encouragement, and insightful conversations during this work. 
DH was partially supported during the present work by Office of Naval Research (ONR) grant award N00014-22-1-2082, Stochastic Parameterisation of Ocean Turbulence for Observational Networks. DH and MKS were  partially supported  by European Research Council (ERC) Synergy grant Stochastic Transport in Upper Ocean Dynamics (STUOD) -- DLV-856408. OS acknowledges funding for a research fellowship from Quadrature Climate Foundation, which has supported his contribution to this project.

\bibliographystyle{apalike}
\bibliography{main}

\end{document}